\documentclass[10pt,twocolumn,twoside]{IEEEtran}

\IEEEoverridecommandlockouts                              

\usepackage{amsmath}

\usepackage{amsthm}
\usepackage{amssymb}
\usepackage{accents}
\usepackage{graphics} 
\usepackage{epstopdf} 
\usepackage{epsfig} 
\usepackage[labelfont=footnotesize,textfont=footnotesize]{subcaption}
\usepackage[labelfont=footnotesize,textfont=footnotesize]{caption}
\usepackage{color}
\usepackage{calc}
\usepackage{cite}
\usepackage[ruled,vlined,linesnumbered]{algorithm2e}
\SetKwFor{For}{for}{}{end}
\SetKwFor{While}{while}{}{end}

\newtheorem{theorem}{Theorem}
\newtheorem{definition}{Definition}
\newtheorem{lemma}{Lemma}
\newtheorem{proposition}{Proposition}

\newtheorem{assumption}{Assumption}

\newtheorem{corollary}{Corollary}

\title{\LARGE \bf
	Resilient Distributed Averaging
}

\author{Mostafa Safi and Seyed Mehran Dibaji
		\thanks{Mostafa Safi is with the Department of Aerospace Engineering, Amirkabir University of Technology, Tehran, Iran,
		{\tt\small halebi@aut.ac.ir, mostafa.safi.office@gmail.com}}%
	\thanks{Seyed Mehran Dibaji is with the Department Mechanical Engineering, Massachusetts Institute of Technolgy, Cambridge, MA, USA,
		{\tt\small dibaji@mit.edu}}
			
}

\begin{document}

	\maketitle

	\begin{abstract}
		In this paper, a fully distributed averaging algorithm in the presence of adversarial Byzantine agents is proposed. The algorithm is based on a resilient retrieval procedure, where all non-Byzantine nodes send their own initial values and retrieve those of other agents. We establish that the convergence of the proposed algorithm relies on \textit{strong robustness} of the graph for locally bounded adversaries. A topology analysis in terms of time complexity and relation between connectivity metrics is also presented. Simulation results are provided to verify the effectiveness of the proposed algorithms under prescribed graph conditions.
	\end{abstract}

%
\IEEEpeerreviewmaketitle
	
	\section{Introduction}\label{Sec: Intro}
	\IEEEPARstart{A}{fter} Stuxnet and the Ukraine power system cyber-attacks, there is a growing awareness for the need to incorporate methods from systems and control in the area of secure and
	resilient control of cyber-physical systems with several promising directions and
	preliminary results. Due to increasing
	connectivity and necessity of wide-area communications, such systems are subject to a range of well-planned and deliberate attacks that might take over the signals and lead
	into irrecoverable and irreparable harmful effects on the physical infrastructures \cite{dibaji2019systems}. 
	
	As a special case of cyber-physical systems, multi-agent systems, which refer to networks of dynamic agents interacting locally with each other to achieve global goals, received considerable attentions in the last two decades \cite{moreau2005stability,jadbabaie2003coordination,cai2012average,
	kashyap2007quantized,dibaji2018TAC}. One of the key problems in multi-agent systems is the so-called consensus, where the objective is that all agents agree upon a common value. As a particular example, average consensus, where the consensus value is the average of the initial values of the agents, has been investigated for undirected graphs \cite{jadbabaie2003coordination} and directed graphs \cite{cai2012average} based on deterministic or randomized update rules \cite{kashyap2007quantized}. There are applications where averaging plays a crucial role. Note that for achieving averaging, the agents must maintain the sum of their values to be constant over time. For example, they represent the total energy in power systems in \cite{Hadjicostis}.

	In this paper, our focus is on security for averaging dynamics. The work \cite{mo2017privacy} investigates the problem of privacy in average consensus settings. Another type of attacks is the so-called \textit{Byzantine} attacks, where some of the agents are hijacked by an outside attacker and send inconsistent and misguiding data to their neighbors in order to avoid other agents reaching consensus. This problem has been solved using installed observers at the agent level and with the expense of high computation costs in \cite{sundaram2011distributed} (see also \cite{pasqualetti2012consensus, meskin2009actuator} for related problems), where each agent, in the presence of bounded number of Byzantine nodes, retrieves the initial values of other agents and detects the faulty agents by making use of alternative paths in the graph to perform the task of the average consensus. The node connectivity of the underlying network is the condition to fully retrieve the initial states, assuming that each non-Byzantine node has a full knowledge of the topology of the graph and significant computational capabilities. 
	
	In the literature, the so-called MSR algorithms have been developed for resilient consensus, where each normal agent eliminates the most deviated agents in the updates. This class of algorithms has been extensively used in computer science \cite{azadmanesh2002asynchronous, bouzid2010optimal, vaidya2012iterative} as well as control \cite{leblanc2013resilient, dibaji2015consensus}. In this paper, we have a different approach to solve the average consensus problem. Based on Certified Propagation Algorithm (CPA), we propose a retrieval procedure for the purpose of average consensus in which, in contrast with \cite{sundaram2011distributed}, the global knowledge of network topology and high computational capabilities of each agent are not required. One reason for this drastic change is that the retrieval process does not involve the agents dynamics. However, the cost of such fully distributed algorithms arises in the more restrictive topology which will be shown in terms of strongly robust graphs as part of the convergence condition. In fact, we use the notion of graph robustness \cite{zhang2012robustness} to ensure that regular nodes achieve consensus in a distributed manner, where each regular node is affected by at most $f$ Byzantine node among its neighbors and does not have access to any global knowledge about topology except the value of $f$.

 Convergence of CPA in the presence of Byzantine adversaries has been also discussed under other topological conditions rather than strong robustness \cite{koo2004broadcast,pelc2005locallybounded,bhandari2009reliable,
	ichimura2010new,zhang2012robustness,tseng2015broadcast}. The most recent work in this area proposes a condition which is based on partitioning the nodes into regular and adversarial nodes \cite{tseng2015broadcast}. As a result, checking the condition can be done only by the blind assumption that any induced subgraph could be adversarial. We show that the time complexity for checking this condition in large scale networks is significantly higher than the case of strongly robust graphs. 
	
	Moreover, while \cite{zhang2012robustness} generally considers resiliency of information diffusion in a network against \textit{Malicious} agents, we propose a distributed algorithm which solves averaging problem in the presence of Byzantine agents\footnote{In computer science literature, Byzantine nodes are capable to send various false values to each of their neighbors, while malicious nodes can only broadcast a single faulty value to all the neighbors.}. Our algorithms are also capable to detect the adversarial nodes using an extra memory. Considering asynchrony and delays in communications and discussion on connectivity notions of a graph are also parts of our contribution in this paper. Based on the realistic assumptions we set, our algorithms and update rules are easy to implement in practice. Furthermore, although we focus on the distributed average consensus, our retrieval method can be utilized to solve more general purpose problems such as resilient distributed function calculation without any modification. 
	
	The rest of the paper is as follows. In Section~\ref{sec.pre}, we state the preliminaries and problem setup. Section~\ref{Sect: MainSection} is devoted to the main results in distributed resilient averaging over synchronous and asynchronous updates. A time complexity analysis and relation between different connectivity metrics are provided in Section~\ref{sec.topology}. Simulations are put in Section \ref{Sect: SimulationExample}. Finally, Section \ref{Sec: Conclusion} concludes the paper. This paper is the extended version of \cite{dibaji2019resilient}. We provide some new results and the eliminated proofs in Section \ref{Sect: MainSection} and add the whole Section~\ref{sec.topology}.
	\section{Preliminaries and Problem Statement} \label{sec.pre}
	
	\subsection{Graph Theory}
	A digraph is represented by $\mathcal{D} = (\mathcal{V},\mathcal{E})$, where the set of nodes and edges are represented by $\mathcal{V} = \{ 1, \ldots ,N \}$ and $\mathcal{E} \subseteq \mathcal{V} \times \mathcal{V}$, respectively. Accordingly, the graph $\mathcal{D}_m = (\mathcal{V}_m,\mathcal{E}_m)$ is a subgraph of $\mathcal{D}$ if $\mathcal{V}_m \subseteq \mathcal{V}$ and $\mathcal{E}_m \subseteq \mathcal{E}$. An induced subgraph is such that $\mathcal{E}_m =\mathcal{E} \cap (\mathcal{V}_m \times \mathcal{V}_m)$. An edge from node $j$ pointing to node $i$ implies data transmission from node $j$ to node $i$ and is denoted by $(j,i)$. The set of incoming and outgoing neighbors of node $i$ are the set $\mathcal{N}_i^- = \{j \vert (j,i) \in \mathcal{E} \}$ and $\mathcal{N}_i^+ = \{j \vert (i,j) \in \mathcal{E} \}$, respectively. A node $j$ is said to be an outgoing neighbour of node $i$ if $(i,j) \in \mathcal{E}$. Also, $\mathcal{D}$ may be undirected if all of its edges are bidirectional, i.e. if $(i,j) \in \mathcal{E}$, then $(j,i) \in \mathcal{E}$, too. For a given $\mathcal{D}$, its underlying graph $\mathcal{G}_\mathcal{D}$ is the one obtained by removing the directions of all edges in $\mathcal{D}$. The number of incoming edges of node $i$ is denoted by $d_{\text{in}}(i)$. A complete graph of order $N$ is a graph in which $(i,j) \in \mathcal{E}$ for all $i,j \in \mathcal{V}$ and is denoted by $\mathcal{K}_N$. A path is a sequence $(v_1,v_2, \ldots ,v_p)$ in which $(v_i,v_{i+1}) \in \mathcal{E}$, where $p>1$ and $i=1, \ldots ,p-1$. Two paths from $v_1$ to $v_p$ are node-disjoint, if they do not share any vertices other than $v_1$ and $v_p$.
	
	In our paper, the key topological notion is the so-called strong robustness. Robustness is a connectivity measure of digraphs, which has been previously used in the literature of resilient distributed computations over networks \cite{zhang2012robustness,mitra2018secure}. The notion of robust graphs is defined as follows.
	
	\begin{definition} \rm \label{df.reach}
		($r$-reachable set) For a digraph $\mathcal{D}$, a subset $\mathcal{S}$ of its nodes is said to be an $r$-reachable set if $\exists i \in \mathcal{S}$ such that $\vert \mathcal{N}_i^- \setminus \mathcal{S} \vert \geq r$, where $r \in \mathbb{Z}_{\geq 1}$.
	\end{definition}
	
	\begin{definition} \rm \label{df.rob}
		($r$-robust graph) A graph $\mathcal{D}$ is $r$-robust if for every pair of nonempty, disjoint subsets of $\mathcal{V}$, at least one of the subsets is $r$-reachable, where $r \in \mathbb{Z}_{\geq 1}$.
	\end{definition}
	
	Moreover, there are other variants of robust graphs, used in \cite{zhang2012robustness} and \cite{mitra2018secure}, for analyzing the resiliency of a digraph against adversarial nodes.
	
	\begin{definition}  \rm \label{df.Srobw}
		(Strongly $r$-robust graph w.r.t. $\mathcal{S}$) For a digraph, a set of nodes $\mathcal{S} \subset \mathcal{V}$ and $r \in \mathbb{Z}_{\geq 1}, r \leq \vert \mathcal{S} \vert$, we say that $\mathcal{D}$ is strongly $r$-robust with respect to $\mathcal{S}$, if for any nonempty subset $\mathcal{C} \subseteq \mathcal{V} \setminus \mathcal{S}$, $\mathcal{C}$ is $r$-reachable. 
	\end{definition}
	
	\begin{definition}  \rm \label{df.Srob}
		(Strongly $r$-robust graph) A digraph $\mathcal{D}$ is strongly $r$-robust if for any nonempty subset $\mathcal{S} \subseteq \mathcal{V}$, either $\mathcal{S}$ is $r$-reachable or $\exists i \in \mathcal{S}$ such that $\mathcal{V} \setminus \mathcal{S} \subseteq \mathcal{N}_i^-$, where $r \in \mathbb{Z}_{\geq 1}$ and $r \leq \lceil N/2 \rceil$.
	\end{definition}

We also use the original definition of connectivity from \cite{geller1971connectivity}. For the connectivity of a graph, there are two ways to measure them: One is node connectivity and the other is edge connectivity \cite{gross2005graph}. Here, we focus on the former notion. 

For digraphs, as a variety of connectivity notions can be defined, node-connectivity is a much more complicated measure. The following definition presents these categories and introduces the connecitivity measures thoroughly.
\begin{definition} \rm \label{df.con}
(Connectivity) A digraph $\mathcal{D}$ is \textit{strongly connected} if 
for every pair of nodes, they are mutually reachable; it is \textit{unilaterally} connected if every node can either reach or be reached from every other node; and it is \textit{weakly} connected if the underlying graph $\mathcal{G}_\mathcal{D}$ is connected. Digraphs which are not in any of the aforementioned groups are \textit{disconnected}. 

Note that each digraph $\mathcal{D}$ belongs to one of the connectivity categories $C_i, i=0,1,2,3$, defined as follows:

\centering
\begin{tabular}{p{50pt}p{145pt}}
	  \parbox{70pt}{\vspace{5pt}\textbf{Category}} 
	& \parbox{200pt}{\vspace{5pt}\textbf{Type of Graphs}}\\
	  \hline
	  \parbox{70pt}{\vspace{5pt} $~~~~~C_0$}
	& \parbox{200pt}{\vspace{5pt} disconnected}\\
	  \parbox{70pt}{\vspace{5pt} $~~~~~C_1$}
	& \parbox{200pt}{\vspace{3.5pt} weak but not unilateral}\\
	  \parbox{70pt}{\vspace{5pt} $~~~~~C_2$}
	& \parbox{200pt}{\vspace{5.5pt} unilateral but not strongly connected}\\
	  \parbox{70pt}{\vspace{5pt} $~~~~~C_3$}
	& \parbox{200pt}{\vspace{5.5pt} strongly connected}\\
\end{tabular}
\end{definition}
As the above categories might suggest, the connectivity notions for digraphs are complicated. Also, the \textit{$ij$ connectivity} $\kappa_{ij}$ of a digraph $\mathcal{D}$ is the minimum number of nodes whose removal changes $\mathcal{D}$ from a digraph in $C_i$ to one in $C_j$. 
A special case is when $j=0$, in which case, we set $\kappa_{ij}=N-1$. Thus, for $\mathcal{D}$ in $C_i$, $i>0$, it is trivial that $\kappa_{i0} = \kappa_{\mathcal{G}_\mathcal{D}}$. In particular, we define the \textit{strong connectivity} of a strongly connected digraph $\mathcal{D}$ as $\kappa_3(\mathcal{D})=\min \{\kappa_{30},\kappa_{31}, \kappa_{32} \}$. Thus $\kappa_3 (\mathcal{D})$ is the  minimum number of nodes whose removal renders $\mathcal{D}$ non-strong or trivial. Similar definitions of $\kappa_0$, $\kappa_1$, $\kappa_2$ can be introduced.

	
	In Section \ref{sec. grrob}, we provide detailed discussions on the relation among robustness and connectivity of a graph. 
	
	\subsection{Distributed Average Consensus}
	In a given network $\mathcal{D}$ of $N$ agents, at time instant $k$, each node $i \in \mathcal{V}$ has a scalar state $x_i[k]$. The \textit{average consensus} problem considers designing distributed algorithms by which the nodes update their states using only the local information of their neighbour nodes such that all $x_i[k]$ eventually converge to the initial average $x_a = \sum_{i=1}^N x_i[0]/N$. The following definition formulates the objective of the distributed consensus averaging scheme.
	
	\begin{definition} \rm \label{df.ave}
		(Average consensus) A network of $N$ agents is said to achieve average consensus if for every initial condition $x_i[0], i=1, \ldots, N$, it holds that $\lim_{k \to \infty} x_i[k] = x_a, \forall i \in \{ 1, 2, \ldots , N \}$.
	\end{definition} 
\subsection{Adversarial Model}
	We consider a problem where the network has to deal with Byzantine adversaries. A Byzantine adversarial node can possess complete knowledge about graph topology and all the communications between the nodes at every time step. It also can deviate from the rules of any prescribed algorithm in arbitrary ways, and can transmit different state values to different neighbours at the same time step.
	
	It is apparent that no distributed consensus algorithm would succeed if too many nodes are adversarial. We partition the set of nodes $\mathcal{V}$ into two subsets: The set $\mathcal{R}$ of regular nodes and the set $\mathcal{A} = \mathcal{V} \setminus \mathcal{R}$ of adversarial nodes. In the literature dealing with distributed fault-tolerant algorithms, it is a common assumption to assign an upper bound $f$ to the total number of adversarial nodes in the network. This is known as the $f$-total adversarial model. However, to allow for a large number of adversaries in large scale networks, we consider a locally bounded fault model, taken from \cite{koo2004broadcast},\cite{pelc2005locallybounded}, defined as follows.
	
	\begin{definition} \rm \label{df.flmod}
		($f$-local adversarial model) A set $\mathcal{A}$ of adversarial nodes is $f$-locally bounded if it contains at most $f$ adversarial nodes in the set of neighbors of each regular node, i.e. $\vert \mathcal{N}_i^- \cap \mathcal{A} \vert \leq f, \forall i \in \mathcal{R}$.
	\end{definition}
	\begin{definition} \rm \label{df.flmod}
		($f$-total adversarial model) A set $\mathcal{A}$ of adversarial nodes is $f$-totally bounded if it contains at most $f$ adversarial nodes in the underlying graph of a network, i.e. $\vert \mathcal{V} \cap \mathcal{A} \vert \leq f$.
	\end{definition}

	\vspace{-.4cm}
	\subsection{Problem Statement}
	In this paper, the aim is to propose a distributed strategy to solve the average consensus problem in the presence of Byzantine adversaries formally put forth as follows.
	
	\begin{definition} \rm \label{df.resavrg}
		(Resilient average consensus) A multi-agent system over the graph $\mathcal{D}$ under Byzantine adversarial attacks is said to achieve resilient average consensus if for every initial value $x_i[0] \in \mathcal{I}, i=1, \ldots, N$, it holds that $\lim_{k \to \infty} x_i[k] = x_a, \forall i \in \mathcal{R}$, where $x_a = \sum_{i \in \mathcal{R}} x_i[0] /\vert \mathcal{R} \vert$ and $\mathcal{I} \subset \mathbb{R}$ is an interval of safe initial values.
	\end{definition}
	
	We note here that the average value to be calculated by regular nodes is that of all agents excluding the adversarial nodes. In fact, we require the regular nodes to arrive at $x_a$. Accordingly, any nodes broadcast initial values within $\mathcal{I}$ are assumed to be regular, unless, some decide to change their values in the course averaging process. In other words, adversarial nodes may decide to broadcast initial values out of $\mathcal{I}$, or to change their values after the retrieval started.
		
	It is also noteworthy that the application of our fully distributed algorithms is not limited to resilient averaging and can be used for the general problem of \textit{function calculation}, where regular nodes compute a given function of the initial values of the network in a distributed manner \cite{sundaram2011distributed}. 
	
	To present a solution for the resilient average consensus problem, we first introduce a novel distributed algorithm for securely accepting and broadcasting information through a network in the presence of $f$-local Byzantine adversaries. Then, we analyze the required constraints on the graph topology, which guarantee achieving resilient average consensus by all of the regular nodes of the network.
\vspace{-.2cm}
\section{Resilient Distributed Averaging}\label{Sect: MainSection}
In this section, we present our strategy that each regular node has to follow in accepting and broadcasting the values it receives from its neighbours and the update rule for its own state value to achieve average consensus. We describe our solution in two steps: the resilient distributed retrieval of initial state values and the averaging rule each node has to execute. The following definition is a formal statement of the resilient distributed retrieval in a given network.

\begin{definition} \rm \label{df.ret}
	(Resilient distributed retrieval) A network of $N$ nodes under Byzantine attacks is said to achieve resilient distributed retrieval if each node $i \in \mathcal{R}$ can retrieve the initial values of all the other regular nodes, i.e. $x_j[0]$, $j \in \mathcal{R} \setminus \{ i\}$. 
\end{definition}
\vspace{-.4cm}
\subsection{Resilient Distributed Retrieval}
To reach the average value of the network, each regular node must obtain the initial state values of other regular nodes. Inspired by the CPA, we propose the Secure Accepting and Broadcasting Algorithm (SABA) for node $i \in \mathcal{R}$ in the presence of $f$-locally bounded Byzantine adversarial nodes. With this algorithm, regular nodes can securely identify and accept the true initial state values of other regular nodes and broadcast them through the network. 

To this end, each regular node $i$ uses a persistent memory vector $m^i[k]=[m_1^i[k], \ldots , m_{\bar{N}}^i[k]]$, $\bar{N} \geq N$, to record the values received and accepted from the incoming neighbors in $\mathcal{N}_i^-$ at each time instant $k \geq 0$. The element $m_n^i[k]$, $n \in \{1, \ldots , N\}$ of the vector is associated to node $n$'s state value in node $i$'s memory. At $k=0$, the memory vector is created as $m^i[0]=[ \ ]_{1 \times \bar{N}}$, where $[ \ ]$ is an empty vector. Here, we assume that regular nodes know an upper bound for the number of network's nodes $\bar{N} \geq N$. Note that this assumption is not restrictive and means that the regular nodes need not know the exact value $N$ for their updates\footnote{Using variable-sized memories, this assumption may not be required. In fact, we define the upper bound $\bar{N}$ to avoid using variable-sized memories and used fixed-sized static memories instead, which are preferred in implementations.}.  

Each regular node $i$ begins to execute the algorithm at $k=0$ and sets $m_i^i[0] = x_i[0]$. At $k=1$, it broadcasts $m^i[0]$ to all of its outgoing neighbors, receives the initial state values of its incoming neighbors, and updates its memory, i.e. $m_j^i[1] = m_j^j[0]$, $j \in \mathcal{N}_i^-$. In this subsection, we assume that there are no communication delays in the network; each regular node $i$ at time instant $k$ simultaneously sends and receives the data packets related to time instant $k-1$. Also, as we discussed earlier, regular nodes are supposed to send initial values within $\mathcal{I}$. Adversarial nodes may send initial values out of the safe interval $\mathcal{I}$, change their values in the course of retrieval, or send different values to each of their neighbors. In the first two cases, the adversarial nodes are easily detectable. The latter case will be prevented by a majority voting: Each regular node $i$ at time instant $k>1$ only accepts the values that are sent by more than $f+1$ incoming neighbors and saves such values with the corresponding label tag $n$ in $m_n^i[k]$. As the algorithm is running, each regular node $i$ fills its memory $m^i[k]$ with more initial values tagged with new labels. The pseudo-code of SABA can be seen in Algorithm~\ref{alg1}.

Next, we investigate the required network topology under which all regular nodes can retrieve the initial state values of other regular nodes by running the SABA. Note that the regular nodes do not know when to stop executing the algorithm. Thus, they are programmed to run the SABA for a lower bounded number of steps denoted by $\bar{K}$.

\begin{algorithm}[t]
\footnotesize
	\SetAlgoLined
	\textbf{Initialization}\\
	The regular node $i$ creates a persistent empty memory $m^i[0]=[ \ ]_{1 \times \bar{N}}$, where $\bar{N} \geq N$, and sets $m_i^i[0] = x_i[0]$.\\
	\If{$k=1$}{
		The regular node $i$ broadcasts its memory vector $m^i[0]$ to all its neighbors.\\
		\For{$j \in \mathcal{N}_i^-$}{
			The regular node $i$ receives $m^j[0]$ and updates its memory $m_j^i[1]=m_j^j[0]$, $j \in \mathcal{N}_i^-$.
		}
	}
	\If{$k>1$}{
		The regular node $i$ broadcasts its memory vector $m^i[k-1]$ to all its neighbors.\\
		\For{$n \in \{1,2,\ldots , \bar{N} \}$}{
			If the regular node $i$ received an identical value $m_n^j[k-1]$, $j \in \mathcal{N}_i^-$ from $f + 1$ incoming neighbors, then it accepts that value and saves it in the memory $m_n^i[k]$.
		}
	}
	\hrulefill \\
	\KwResult{$ m^i[k]=[m_1^i[k], \ldots , m_{\bar{N}}^i [k]]$, $\bar{N} \geq N$	
	}
	\caption{Secure Accepting and Broadcasting Algorithm (SABA)} 
	\label{alg1}
\end{algorithm}

\begin{theorem} \rm \label{th.alg1}
	Each node $i \in \mathcal{R}$ in the network $\mathcal{D}$, by executing the SABA for $\bar{K} \geq 2N-1$ steps\footnote{$N$ in this lower bound can be replaced with $\bar{N}$, which is an estimate of the number of nodes, to make the stopping time also independant from the global information of the network.}, will retrieve $x_\ell[0]$, $\ell \in \mathcal{R} \setminus \{i \}$ under the $f$-local adversarial model if $\mathcal{D}$ is strongly $(2f + 1)$-robust.
\end{theorem} 

\begin{proof}
	Suppose that there is a finite time $\bar{K}$ such that by that time, all the regular nodes would have received all the initial state values  that are not faulty. Consider node $\ell \in \mathcal{R}$. Each regular node $i \in \mathcal{N}_\ell$ receives the initial state value $x_\ell [0]$ directly. We use contradiction to prove that all other nodes will receive $x_\ell [0]$. Let $\mathcal{U}$ denote the set of all the regular nodes which cannot receive $x_\ell [0]$. According to Def.~\ref{df.Srob}, we know that some node $i \in \mathcal{U}$ must have $2f + 1$ neighbors outside $\mathcal{U}$. At most $f$ of these nodes can be adversarial and all other nodes are regular nodes that have received $x_\ell [0]$ and re-broadcasted it at some time step $k \leq \bar{K}$. This contradicts the assumption that node $i$ fails to get the initial state value of node $\ell$. The same argument holds for other regular nodes. Thus, all the regular nodes securely access $x_i[0]$, $i \in \mathcal{R}$.

	To find a lower bound for $\bar{K}$, consider node $i,j \in \mathcal{V}$ and assume that $x_j[0] \in \mathcal{I}$. According to the SABA, $x_j[0]$ must retrieved by $i$. In the worst case, $x_j[0]$ has to be passed by all the nodes in $\mathcal{R} \setminus \{ i,j \}$ to the node $i$, which takes $N-1$ time steps. However, consider that node $j$ is an adversarial node and decides to change its value at $k=N-1$. Then, it takes another $N-1$ time steps that $x_j[N-1]$ retrieved by node $i$ in the worst case. Thus, to ensure that all initial values retrieved by all the nodes, SABA must be executed for $2N-1$ time steps including $k=0$, i.e. $\bar{K} = 2N-1$. Note that each node $i$ does not know exactly when to stop executing the SABA, and thus, keeps executing it until $\bar{K}$ to ensure that all the initial state values are retrieved securely and correctly.
\end{proof}

%
\vspace{-.1cm}
Interestingly, by executing the SABA up to $\bar{K}$ steps, all the regular nodes can detect adversarial nodes even if they broadcast values in the safe interval $\mathcal{I}$ to deceive the regular nodes. The following scenarios are possible:

\begin{itemize}
	\item[i)] An adversarial node $s$ may broadcast a false constant state value $m_n^s[k]=a \in \mathcal{I}$ for the label $n \in \{1,2,\ldots,\bar{N} \}$, from a time instant $t=k>0$ to node $i$. According to Theorem~\ref{th.alg1}, the regular node $i$ will at some point $t$ (before $\bar{K}$) receive the value $m_n^j[k] \neq a$, $j \in \mathcal{N}_i^-$ from at least $f+1$ regular neighbors. Then, it will easily detect node $s$ as an adversarial node.
	
	\item[ii)] An adversarial node $s$, at time instant $t=k$, may broadcast $m_n^s[k]$ for the label $n \in \{1,2,\ldots,\bar{N} \}$ to node $i$ and change it at some time $t>k$. According to the SABA, all the regular nodes are supposed to constantly send the accepted values up to $\bar{K}$ steps. Thus, node $i$ simply identifies $s$ as an adversarial node. However, note that node $i$ needs extra memory buffers for each of its neighbors to track the changes in the values they broadcast while the algorithm is executed.
\end{itemize}

While the SABA is inspired by the CPA \cite{zhang2012robustness}, our retrieval algorithm considers information diffusion for all nodes of the graph simultaneously. Accordingly, the graph condition needed for the SABA to succeed is inclusive with respect to the graph condition associated with the CPA -- each graph which is strongly $r$-robust with respect to each of its nodes is in fact strongly $r$-robust (see Section~\ref{sec.topology} for a brighter insight). Also, the SABA can be easily developed for the asynchronous settings with communication delays (see Algorithm~\ref{alg2}). Furthermore, the most recent work in the literature,  \cite{tseng2015broadcast}, proposed a condition for convergence of the CPA in the presence of $f$-locally bounded adversaries. However, this condition imposes a higher order of time complexity with respect to strong robustness if we want to use it for our retrieval strategy (see Subsection~\ref{sec.cox} for more details). Therefore, we use the strong robustness notion as the sufficient condition for the convergence of our algorithms, while the condition in \cite{tseng2015broadcast} can be our necessary condition.\\
\vspace{-.5cm}
\subsection{Synchronous Averaging}
In this subsection, we suppose that the network makes updates synchronously and propose an update rule for the regular nodes to update their states using the received and accepted initial state values at each time step, converging to the average consensus asymptotically.

The regular node $i$ by running the SABA updates its memory $m^i[k]$, at each time step $k$, and obtains more initial state values of other regular nodes which are received over time from its neighbors. We define the instantly cumulative average of the initial state values received by node $i$ up to the time instant $k$ as follows:

\begin{equation} \label{eq.qaver}
\phi_i[k]=\frac{\sum m_n^i[k]}{\lambda[k]}, \ n \in \mathbb{M}^i[k],
\end{equation}
where $\mathbb{M}^i[k]$ is the set of indices of the elements in the memory vector $m^i[k]$ which are nonempty and its cardinality is given by $\lambda[k] = \vert \mathbb{M}^i[k] \vert$. 
Then, to handle any possible fluctuation in the state updates, the regular node $i$ utilizes a low pass filter (also known as exponential smoothing \cite{brown2004smoothing}) and updates its state as:

\begin{equation}  \label{eq.upd}
\begin{aligned} 
x_i[0]&= \phi_i[0],\\
x_i[k]&=\epsilon_i x_i[k-1] + (1-\epsilon_i) \phi_i[k], \ \forall k>0,
\end{aligned}
\end{equation}
where $0 \leq \epsilon_i<1$ is the filter gain. In the case $\epsilon_i = 0$, the state value of $x_i[k]$ is equal to the instantly cumulative average $\phi_i[k]$. The filter gain $\epsilon_i$ can be chosen arbitrarily, but its size may determine the susceptibility of the node's dynamics to the instant changes in the state values. However, the filter will impose delays in convergence of $x_i$. 


Here, we connect our results together to show how our proposed strategy leads a network with specific topology constraint to resilient average consensus. 

\begin{theorem} \rm \label{th.sum}
	Each node $i \in \mathcal{R}$ achieves average consensus by executing the SABA for $\bar{K} \geq N-2$ steps and performing the update rule (\ref{eq.upd}) under the $f$-local adversarial model if the network is strongly $(2f + 1)$- robust.
\end{theorem}

\begin{proof}
	Referring to Theorem~\ref{th.alg1}, each node $i \in \mathcal{R}$ in the network $\mathcal{D}$, which is strongly ($2f+1$)-robust, will securely retrieve $x_i[0]$, $i \in \mathcal{R}$, if it executes the SABA for $\bar{K}$ steps. Then, according to (\ref{eq.qaver}), $\phi_i[k]$ is a linear combination of the received and accepted initial state values $x_i[0], \ i \in \mathcal{R}$, at each time instant $k$ and converges to $x_a$ asymptotically as each node $i$ will ultimately receive the initial state values of other regular nodes at $t=\bar{K}$. Now, we add and subtract $2x_a$ on both sides of (\ref{eq.upd}) and rewrite it as $(x_i[k]-x_a) - (\phi_i[k]-x_a)=\epsilon_i (x_i[k-1]-x_a) - \epsilon_i (\phi_i [k]- x_a)$, which leads to the following relation if $k \to \infty$ (since $\lim_{k \to \infty} \phi_i[k] = x_a$): $x_i[k]-x_a=\epsilon_i (x_i[k-1]-x_a)$. Since $0 \leq \epsilon_i<1$, this system is Schur stable and $\lim_{k \to \infty}  x_i[k]=x_a$, $\forall i \in \mathcal{R}$.
\end{proof}

In view of the literature of resilient averaging, we achieved a significant contribution in terms of distributed decision making in the strategy, the memory usage, and the computational capability of each regular node. In \cite{sundaram2011distributed}, the required topology constraint for achieving resilient averaging proposed to be ($2f+1$)-connectivity. However, the method is not distributed since each regular node needs to know the full network topology (in particular, the observability and invertibility matrices, both of which are functions of the adjacency matrix of the network) to update its state value and to achieve average consensus. By contrast, our design is fully distributed although it imposes a more complicated topology constraint on the network, that is, the network must be strongly ($2f+1$)-robust. Furthermore, the proposed strategy in \cite{sundaram2011distributed} forced the regular nodes to utilize a large amount of memory for storing two large matrices at each step and high computational capabilities for calculating the rank of matrices or multiplying matrices. On the other hand, our presented strategy requires each regular node to utilize only $\bar{N} \geq N$ memory for achieving average consensus in a distributed manner and dedicate $\bar{N}$ extra memory buffers for each of its neighbors to detect and eliminate all the adversarial nodes.

Theorem~\ref{th.sum} gives a sufficient, and not necessary, condition for secure convergence under $f$-local adversarial model. We propose a necessary condition for achieving resilient average consensus using SABA and update rule (\ref{eq.upd}) which is valid for both $f$-local and $f$-total adversarial models. Relying on Theorem~\ref{th.sum}, in the rest of the paper we seek resilient retrieval which can conclude the resilient averaging.

\begin{theorem} \rm \label{th.necrob}
If a network achieves resilient average consensus by performing SABA and update rule (\ref{eq.upd}) under $f$-total \slash$f$-local adversarial model, its underlying graph is $(2f+1)$-robust.
\end{theorem}

\begin{proof}
Consider a network $\mathcal{D}$ that reaches average consensus under $f$-total adversarial model by executing SABA and its underlying graph is not $(2f+1)$-robust. Thus, there exist two distinct subsets $\mathcal{S}_1, \mathcal{S}_2 \subset \mathcal{V}$ which both are at most $2f$-reachable. Suppose that node $i$ and $j$ are in $\mathcal{S}_1$ and $\mathcal{S}_2$, respectively, and all $f$ Byzantine nodes are in $\mathcal{V} \setminus (\mathcal{S}_1 \cup \mathcal{S}_2)$. Node $i$, to calculate the average consensus using update rule (\ref{eq.upd}), needs the initial state value of node $j$. However, it has at most $2f$ incoming edges from out of $\mathcal{S}_1$. If $f$ of these incoming edges are from $f$ Byzantine nodes and the other $f$ nodes are from the regular nodes out of $\mathcal{S}_1$, then the node $i$ cannot verify the initial state value of node $j$ received from its neighbors out of $\mathcal{S}_1$ by majority voting e.g. the $f$ Byzantine nodes send the value $a$ as the initial state of the node $j$ and the other $f$ regular nodes out of $\mathcal{S}_1$ send the value $b$. Therefore, node $i$ cannot achieve the average consensus which is a contradiction. 
\end{proof}

Furthermore, the following proposition shows that SABA fails in certain networks with a constraint in the connectivity sense which does not meet the same robustness constraint proposed in Theorem~\ref{th.sum}. 

\begin{figure}[t]
\fontsize{7}{10}\selectfont
	\vspace{0cm}
	\def \svgscale{.35}
	\hspace{2.5cm}
\begingroup%
  \makeatletter%
  \providecommand\color[2][]{%
    \errmessage{(Inkscape) Color is used for the text in Inkscape, but the package 'color.sty' is not loaded}%
    \renewcommand\color[2][]{}%
  }%
  \providecommand\transparent[1]{%
    \errmessage{(Inkscape) Transparency is used (non-zero) for the text in Inkscape, but the package 'transparent.sty' is not loaded}%
    \renewcommand\transparent[1]{}%
  }%
  \providecommand\rotatebox[2]{#2}%
  \newcommand*\fsize{\dimexpr\f@size pt\relax}%
  \newcommand*\lineheight[1]{\fontsize{\fsize}{#1\fsize}\selectfont}%
  \ifx\svgwidth\undefined%
    \setlength{\unitlength}{278.24151948bp}%
    \ifx\svgscale\undefined%
      \relax%
    \else%
      \setlength{\unitlength}{\unitlength * \real{\svgscale}}%
    \fi%
  \else%
    \setlength{\unitlength}{\svgwidth}%
  \fi%
  \global\let\svgwidth\undefined%
  \global\let\svgscale\undefined%
  \makeatother%
  \begin{picture}(1,0.81024562)%
    \lineheight{1}%
    \setlength\tabcolsep{0pt}%
    \put(0,0){\includegraphics[width=\unitlength,page=1]{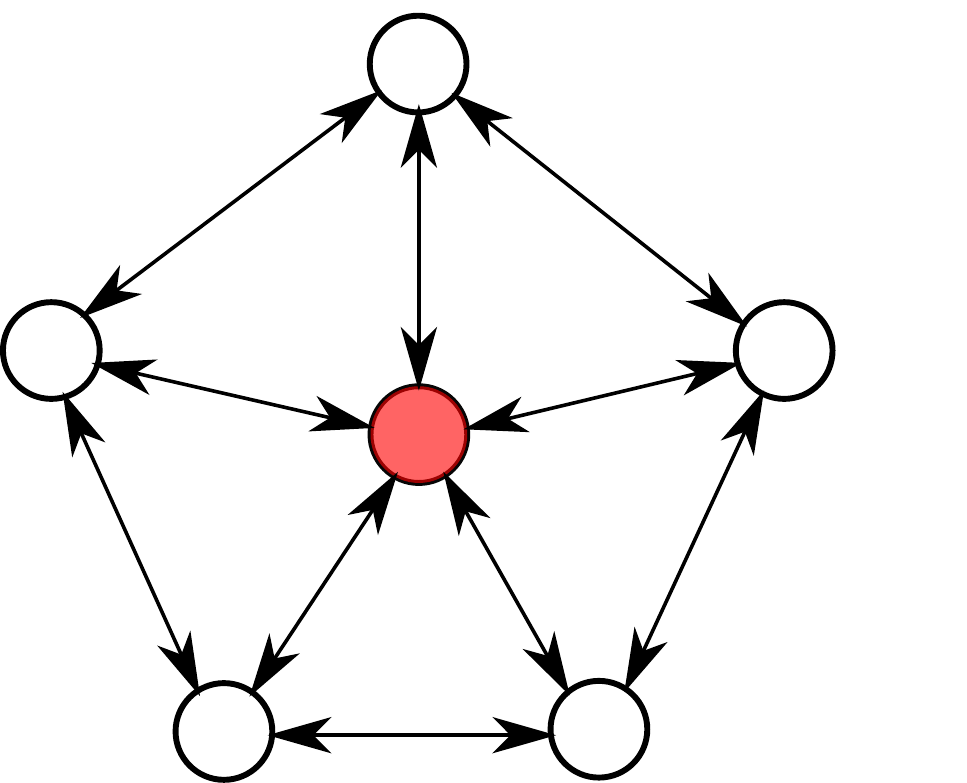}}%
    \put(0.41487183,0.72188377){\color[rgb]{0,0,0}\makebox(0,0)[lt]{\lineheight{1.25}\smash{\begin{tabular}[t]{l}$1$\end{tabular}}}}%
    \put(0.79084615,0.42509137){\color[rgb]{0,0,0}\makebox(0,0)[lt]{\lineheight{1.25}\smash{\begin{tabular}[t]{l}$2$\end{tabular}}}}%
    \put(0.60536297,0.032725){\color[rgb]{0,0,0}\makebox(0,0)[lt]{\lineheight{1.25}\smash{\begin{tabular}[t]{l}$3$\end{tabular}}}}%
    \put(0.21114937,0.03201319){\color[rgb]{0,0,0}\makebox(0,0)[lt]{\lineheight{1.25}\smash{\begin{tabular}[t]{l}$4$\end{tabular}}}}%
    \put(0.03219294,0.42765003){\color[rgb]{0,0,0}\makebox(0,0)[lt]{\lineheight{1.25}\smash{\begin{tabular}[t]{l}$5$\end{tabular}}}}%
    \put(0.41465288,0.33774075){\color[rgb]{0,0,0}\makebox(0,0)[lt]{\lineheight{1.25}\smash{\begin{tabular}[t]{l}$6$\end{tabular}}}}%
  \end{picture}%
\endgroup%

	\vspace{0cm}
	\caption{A wheel graph with the adversarial node in the middle, which is ($2f+1$)-connected but not strongly ($2f+1$)-robust.}\vspace{-2mm} \label{fig1} 
\end{figure}
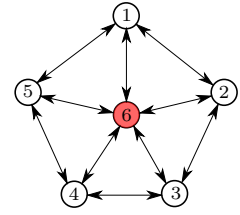 

\begin{proposition} \rm \label{prop.imposible}
There exist networks with ($2f+1$)-connectivity, but not strong ($2f+1$)-robustness such that regular nodes cannot achieve resilient average consensus by performing SABA. 
\end{proposition}

\begin{proof}
	For $f=1$, consider the wheel graph in Fig.~\ref{fig1}, wherein the node 6 is adversarial. It can be proved that wheel graphs are $3$-connected \cite{gross2005graph}. Take node 1 as the source of its initial state value $x_1[0]=a \in \mathcal{I}$, which is fed in its output vector $m_1^1[0]$. At $k=0$, we have $m_1^i[0]=[ \ ]_{1 \times 1}$, $\forall i \neq 1$. At $k=1$, nodes 2 and 5 receive $m_1^1[0]$ and accept it as the true initial value of the node 1, i.e. $m_1^2[1]=m_1^5[1]=a$. At this time instant, the adversarial node 6 decides to send two different values $b,c \in \mathcal{I}$ on behalf of node 1 to nodes 3 and 4, respectively. So, at $k=2$, node 3 receives $m_1^2[1]=a$ from node 2 and $m_1^6[1]=b$ from node 6. Accordingly, node 4 receives $m_1^5[1]=a$ from node 5 and $m_1^6[1]=c$ from node 6. Based on the SABA, nodes 3 and 4 do not have enough data ($f+1=2$ consistent value) and cannot decide which value to accept as the true initial value of the node 1 for all the future time steps. 
\end{proof}

Next, we propose another necessary condition for resilient average consensus of a network under $f$-total adversarial model based on strongly $r$-connectivity. We use an alternative definition of strongly $r$-connected digraphs inferred from Menger's Theorem \cite{west1996introduction} stated as follows:

\begin{lemma} \rm \label{lm.menger}
A digraph $\mathcal{D}$ is strongly $r$-connected if and only if for any pair of nodes $i,j \in \mathcal{V}$ there exist $r$ disjoint paths from $i$ to $j$ in $\mathcal{D}$.
\end{lemma}

\begin{theorem} \rm \label{th.necsuf}
If a network with underlying graph of $\mathcal{D}$ achieves average consensus by executing SABA under $f$-total adversarial model, then $\mathcal{D}$ is strongly $(2f+1)$-connected.
\end{theorem}

\begin{proof}
According to Lemma~\ref{lm.menger}, if $\mathcal{D}$ is strongly $(2f+1)$-connected, there exist at least $2f+1$ internally node-disjoint paths between each pair of nodes or the two nodes are directly connected together. We use contradiction to prove that if a network achieves consensus by executing SABA under $f$-total adversarial model, its underlying graph is strongly $(2f+1)$-robust. Note that accomplishing resilient retrieval is equivalent to achieving average consensus as stated before. Assume that there exist a pair of nodes $i,j \in \mathcal{D}$ that there exist at most $2f$ internally node-disjoint paths between them. Suppose that we have $f$ paths, each has an adversarial node, and $f$ paths that are constructed only by regular nodes. The adversarial nodes can change the initial state value of node $i$ on the path to node $j$. Take the actual initial state value of node $i$ as $x_i[0]=a \in \mathcal{I}$. However, adversarial nodes broadcast $m_n^i[k]=b \in \mathcal{I}$, $n \in \mathcal{A}$ behind of node $i$. In this situation, node $j$ have maximum $f$ initial state values that are consistent. Therefore, it cannot conclude the majority voting and retrieve all the initial state values for averaging.
\end{proof}
\vspace{-.4cm}
\subsection{Asynchronous Averaging}
There may be delays in communications over the network and nodes may update asynchronously. To analyze the effects of asynchrony and delays on our proposed strategy, we set the following assumption on the communications protocol. Then, we propose a modified version of the SABA for asynchronous networks and a slightly different update rule.

\begin{assumption} \rm
	All nodes make, at least, an update within $\bar{k}$ steps and communication delays are upper-bounded by $\bar{\tau}$ steps.
\end{assumption}

Note that in the asynchronous version of the SABA, each node $i \in \mathcal{R}$ utilizes an extra memory buffer for each of its neighbors. It stores the most recently received data packets synchronized by the network clock. The following procedure describes Algorithm~\ref{alg2}. It is referred to as the 
Asynchronous Secure Accepting and Broadcasting Algorithm (ASABA):

\begin{itemize}
	\item[i)] Each node $i \in \mathcal{R}$ expects to find only one initial state value in the packets it receives from each of its incoming neighbors up to time instant $t = \bar{k} + \bar{\tau}$ (this is because each node ensures that it has received data packets from all its neighbours) and updates its memory $m_j^i[k]=m_j^j[k-\widetilde{k}_{ij} - \tau_{ij}]$, where $j \in \mathcal{N}_i^-$, $\tau_{ij} \leq \bar{\tau}$ is the time delay of the last data packet that node $i$ has received from node $j$ (it may be time-varying) and $\widetilde{k}_{ij}$ denotes the time steps elapsed from the time that node $i$ has received the packet of the node $j$ up to the time it makes an update ($\widetilde{k}_{ij} < \bar{k}$).
	
	\item[ii)] Each node $i \in \mathcal{R}$, after accepting the initial state values in the first stage of the algorithm, uses the most recent data packets received from its neighbors $m^j[k-\widetilde{k}_{ij} - \tau_{ij}]$, $j \in \mathcal{N}_i^-$, to update its memory $m^i[k]$ regardless of asynchrony and delays in communications and broadcasts it to all its neighbour.
\end{itemize}

\begin{algorithm}[t]
\footnotesize
	\SetAlgoLined
	\textbf{Initialization}\\
	The regular node $i$ creates a persistent empty memory $m^i[0]=[ \ ]_{1 \times \bar{N}}$, where $\bar{N} \geq N$, and sets $m_i^i[0] = x_i[0]$.\\
	\If{$k \leq \bar{k} + \bar{\tau}$}{
		
		\For{$j \in \mathcal{N}_i^-$}{
			The regular node $i$ broadcasts $m^i[0]$ to its outgoing neighbors $\mathcal{N}_i^+$, takes the last data packet $m^j[k-\widetilde{k}_{ij} - \tau_{ij}]$ in its receiving buffer, and updates its memory as $m_j^i[k]=m_j^j[0]$, $j \in \mathcal{N}_i^-$.
		}
	}
	\If{$k > \bar{k} + \bar{\tau}$}{
		\For{$n \in \{1,2,\ldots , \bar{N} \}$}{
			If the regular node $i$ received an identical value $m_n^j[k-\widetilde{k}_{ij} - \tau_{ij}]$, $j \in \mathcal{N}_i^-$ from $f + 1$ incoming neighbors, then it accepts that value and saves it in the memory $m_n^i[k]$.
		}
		The regular node $i$ broadcasts its memory vector $m^i[k]$ to all its neighbors.\\
	}
	\hrulefill \\
	\KwResult{$ m^i[k]=[m_1^i[k], \ldots , m_{\bar{N}}^i [k]]$, $\bar{N} \geq N$
	}

	\caption{Asynchronous Secure Accepting and Broadcasting Algorithm (ASABA)} 
	\label{alg2}
	
\end{algorithm}

Note that ASABA has to be executed for more time steps than the SABA so that the network retrieves all the initial values. 

\begin{theorem} \rm \label{th.alg2}
	Each node $i \in \mathcal{R}$ in the network $\mathcal{D}$ with communication delays and asynchrony, by executing the ASABA for $\bar{K} \geq (2N-1)(\bar{k}+\bar{\tau})$ steps, will retrieve $x_\ell[0]$, $\ell \in \mathcal{R} \setminus \{i \}$ under the $f$-local adversarial model if $\mathcal{D}$ is strongly $(2f + 1)$-robust.
\end{theorem}

\begin{proof}
	The proof is the same as Theorem~\ref{th.alg1}. We just qualify the argument finding a lower bound for $\bar{K}$ in the presence of communication delays and asynchrony. In the synchronous case, we saw that to ensure that all initial values retrieved by all the nodes, SABA must be executed for $2N-1$ time steps. Each transmission in a network with asynchrony and communication delays can take up to $\bar{k}+\bar{\tau}$ steps in the worst case. Thus, to ensure that the network accomplishes the retrieval of initial state values, each node $i$ has to execute ASABA for $\bar{K} \geq (2N-1)(\bar{k}+\bar{\tau})$ steps.
\end{proof}

Moreover, the general form of update rule (\ref{eq.upd}) for asynchronous network is:

\begin{equation} \label{eq.asupd}
\begin{aligned} 
x_i[0]&= \phi_i[0],\\
x_i[k]&=\epsilon_i x_i[k_i^-] + (1-\epsilon_i) \phi_i[k] \ \forall k>0,
\end{aligned}
\end{equation}
where, $k_i^-$ is the last time before $k$ that the node $i$ made an update at. Note that $\phi_i[k]$ has been updated by the last initial values in $m^i[k]$ which are received up to $t=k$.\\

Note that Theorem~\ref{th.necsuf} is also valid for the asynchronous case as we only used topology concepts in its proof. We also simply generalize Theorem~\ref{th.sum} for an asynchronous network with delays in communications as follows.

\begin{proposition} \rm \label{propo.asynave}
	Each node $i \in \mathcal{R}$ achieves average consensus by executing the ASABA for $\bar{K} \geq (2N-1)(\bar{k}+\bar{\tau})$ steps and performing the update rule (\ref{eq.asupd}) under the $f$-local adversarial model and in the presence of communication delays and asynchrony if the network is strongly $(2f + 1)$- robust.
\end{proposition}
\vspace{-.3cm}
\section{Topology Analysis} \label{sec.topology}
In this section, we analyze strong robustness in terms of time complexity and connectivity and compare it with other topology constraints.
\vspace{-.3cm}
\subsection{Time Complexity} \label{sec.cox}
The most recent work in the literature,  \cite{tseng2015broadcast}, proposed a graph condition for convergence of the CPA in the presence of $f$-locally bounded adversaries.

\begin{definition} \rm \label{df.LMA}
Graph $\mathcal{G}(\mathcal{V}, \mathcal{E})$ is $f$-resilient w.r.t. $s$ if the subsets $\mathcal{A},\mathcal{L},\mathcal{M}$ form a partition of $\mathcal{V}$, such that:
\begin{itemize}
\item[i)] source $s \in \mathcal{L}$
\item[ii)] $\mathcal{M}$ is nonempty
\item[iii)] $\mathcal{A}$ is a $f$-local adversarial set, and at least one of the following statements holds:
\begin{itemize}
\item[a)] There exists a node $ v \in \mathcal{M}$ that has at least $f+1$ distinct incoming neighbors in $\mathcal{L}$, or 
\item[b)] $\mathcal{M}$ contains an outgoing neighbor of $s$, i.e. $\mathcal{N}_s^+ \cap \mathcal{M} \neq \emptyset$. 
\end{itemize}
\end{itemize}
\end{definition} 

According to this condition, the CPA from a source node $s$ is correct under $f$-local adverserial model if and only if graph $\mathcal{G}$ is $f$-resilient w.r.t. $s$. However, this constraint is based on robustness of the induced subgraph of regular nodes with respect to $s$, which means that either Byzantine nodes and the source node must be known beforehand and we check the strong robustness of the induced subgraph of regular nodes with respect to $s$ or we have to exhaustively check the conditions in Definition~\ref{df.LMA} on $\mathcal{G}$ for all possible adverserial subsets and source nodes -- in case, we call it $f$-resilient. The first assumption seems strange, while there is no preliminary detection. The second assumption imposes a higher time complexity with respect to strong robustness. Therefore, we use the strong robustness notion as the sufficient condition for the convergence of our algorithms, while the condition in \cite{tseng2015broadcast} can be our necessary condition. The relation between these two conditions is still an open problem.

Here, we propose Algorithm~\ref{alg.robcheck} and Algorithm~\ref{alg.tightcheck} to calculate the time complexity required for each condition. To measure and compare the time complexities, we mention the following standard notion from computer science literature \cite{knuth1976big}.

\begin{algorithm}[t]
\footnotesize
	\SetAlgoLined
	\textbf{Initialization}\\
	Graph $\mathcal{D}(\mathcal{V,E})$ with $N$ nodes is the input.\\
	\For{$\alpha = 1, \ldots, \sum_{\xi=1}^{N} \binom{N}{\xi}$}{
		Find the nonempty and non-repetitive subset $\mathcal{S}_\alpha \subseteq \mathcal{V}$.\\
		
		\For{$i \in \mathcal{S}_\alpha $}{
		Check the incoming edges of node $i$.\\
			\If{$\vert \mathcal{N}_i^- \setminus \mathcal{S}_\alpha \vert \geq 2f+1$ \rm or $\mathcal{V} \setminus \mathcal{S}_\alpha \subseteq \mathcal{N}_i^-$}{$c(\alpha)=1$,\\
			\textbf{Goto} line 3.}
		}
	}
	\hrulefill \\
	\KwResult{If all elements of $c$ equals to 1, the graph satisfies the Strong Robustness Condition.
	}
	\caption{Strongly $(2f+1)$-robustness check} 
	\label{alg.robcheck}
\end{algorithm}
		
\begin{definition} \rm \label{df.bigo} 
(Worst-case complexity) Consider $S, U, L: \mathbb{R} \rightarrow \mathbb{R}$, where $S$ is the number of steps required for termination of an algorithm. Then, $S(n) \in \mathcal{O}(U(n))$ if there exists $\delta \in \mathbb{R}_{>0}$ and $n_0 \in \mathbb{R}_{>0}$ such that $\vert S(n) \vert \leq \delta U(n)$ for all $n \geq n_0$. We say the worst-case complexity of an algorithm is $\mathcal{O}(U(n))$ if $S(n) \in \mathcal{O}(U(n))$.
\end{definition}

In other words, $U(n)$ is the upper bound of $S(n)$ for $n \geq n_0$. Furthermore, to check each incoming edge of a node or counting number of nodes in a set, we deal with a \textit{test} that is a basic operation in algorithms and has time complexity of $\mathcal{O}(1)$. In what follows, we calculate the upper bound for the number of basic operations required for termination of Algorithm~\ref{alg.robcheck} and Algorithm~\ref{alg.tightcheck}, considering the number of nodes in the graph as the only variable. 

\begin{theorem} \rm \label{th.coxrob}
Consider Algorithm~\ref{alg.robcheck} and digraph $\mathcal{D} = (\mathcal{V}, \mathcal{E})$ with $N$ nodes. The worst case complexity of Algorithm~\ref{alg.robcheck} for checking strong robustness of $\mathcal{D}$ is $\mathcal{O}(N^2 2^N)$.
\end{theorem}

\begin{proof}
We define the function $S_{\ref{alg.robcheck}}(N)$ as the number of basic operations needed to perform Algorithm~\ref{alg.robcheck}, where $N$ is the number of nodes in digraph $\mathcal{D}$. In Algorithm~\ref{alg.robcheck} (line 3-4), number of all nonempty subsets $\mathcal{S} \subseteq \mathcal{V}$ is $\sum_{\xi=1}^{N} \binom{N}{\xi}$, where $\xi = \vert \mathcal{S} \vert$. Consider node $i \in \mathcal{S}$. Incoming edges of node $i$ from the nodes in $\mathcal{V} \setminus \mathcal{S}$ have to be checked (line 5-11). This can be done by $(N-\xi)\xi$ test operations. Therefore, in worst case, we have $S_{\ref{alg.robcheck}}(N) = \sum_{\xi=1}^{N} \binom{N}{\xi}(N-\xi)\xi$ test operations. Since $\binom{N}{\xi}(N-\xi)\xi=N \xi \binom{N-1}{\xi}$ and $S_{\ref{alg.robcheck}}(N)=0$ for $\xi = N$, we conclude $S_{\ref{alg.robcheck}}(N)=N \sum_{\xi=1}^{N-1} \xi \binom{N-1}{\xi}$. Using Binomial Theorem, we conclude that $S_{\ref{alg.robcheck}}(N)=N^2 2^{N-2}-N$. Therefore, the worst case complexity of Algorithm~\ref{alg.robcheck} is $\mathcal{O}(N^2 2^{N})$. 
\end{proof}

\begin{algorithm}[t]
\footnotesize
	\SetAlgoLined
	\textbf{Initialization}\\
	Graph $\mathcal{D}(\mathcal{V,E})$ with $N$ nodes is the input.\\
	\For{$s \in \mathcal{V}$}{
		Consider $s$ as a source node.\\
		\For{$\alpha = 1, \ldots , \sum_{\xi=1}^{N-1} \binom{N-1}{\xi}$}{
			Find the nonempty and non-repetitive subset of adversarial nodes $\mathcal{A}_\alpha \subseteq \mathcal{V} \setminus \{s\}$.\\
			\For{$\beta = 1, \ldots , \sum_{\eta=1}^{N- \vert \mathcal{A}_\alpha \vert -1} \binom{N - \vert \mathcal{A}_\alpha \vert - 1}{\eta}$}{
				Find the nonempty and non-repetitive subset $\mathcal{M}_\beta \subseteq \mathcal{V} \setminus (\mathcal{A}_\alpha \cup \{s\})$. The subsets $\mathcal{A}_\alpha, \mathcal{M}_\beta, \mathcal{L}_\beta$ are partitions of $\mathcal{V}$, where $\mathcal{L}_\beta = \mathcal{V} \setminus (\mathcal{A}_\alpha \cup \mathcal{M}_\beta)$ and	$s \in \mathcal{L}_\beta$.\\
				\For{$i \in \mathcal{M}_\beta$}{
					Check the incoming edges of node $i$.\\
					\If{$ \vert \mathcal{N}_i^- \cap \mathcal{A}_\alpha \vert \leq f$ \rm and $\big( \vert \mathcal{N}_i^- \cap \mathcal{L}_\beta \vert \geq f+1$ \rm or $(i,s)=1 \big)$}{$c(s,\alpha,\beta)=1$,\\
					\textbf{Goto} line 7.}
				}
			}
		}
	}
	\hrulefill \\
	\KwResult{If all elements of $c$ equals to 1, the graph satisfies the $\mathcal{LM}f$-condition.
	}
	\caption{$f$-resiliency check}
	\label{alg.tightcheck}
\end{algorithm}
\vspace{-.3cm} 
\begin{theorem} \rm \label{th.coxtight}
Consider Algorithm~\ref{alg.tightcheck} and digraph $\mathcal{D} = (\mathcal{V}, \mathcal{E})$ with $N$ nodes. The worst-case complexity of Algorithm~\ref{alg.tightcheck} to check $\mathcal{D}$ for $f$-resiliency is $\mathcal{O}(N^{3} 3^N)$.
\end{theorem}

\begin{proof}
First, we break down Algorithm~\ref{alg.tightcheck} into four cascading steps that each requires $S_{\ref{alg.tightcheck},i}(N)$ number of basic operations, $i=1, ...4$. Then, the worst case complexity of Algorithm~\ref{alg.tightcheck} is calculated via total number of operations $S_{\ref{alg.tightcheck}}(N)$ which considers all the cascading steps together. 

\begin{enumerate}
\item (line 3-4) In Definition~\ref{df.LMA}, the source node $s$ is assumed to be known to all the nodes. However, to check the condition (iii) of the definition for a graph in general, node $s$ is unknown. Thus, steps (2) to (4) must be done for the $N$ nodes that means $S_{\ref{alg.tightcheck},1}(N)=N$.

\item (line 5-6) As we do not know the set $\mathcal{A}$, the condition (iii) of Definition~\ref{df.LMA} must be checked for every possible subset of adversarial nodes chosen from $\mathcal{V} \setminus \{i\}$ such that $1 \leq \vert \mathcal{A} \vert \leq N-1$. Note that $\mathcal{A}$ includes all the adversarial nodes of the network that can exceed $f$, however, its $f$-local condition is checked in step (4). This is equivalent to finding all nonempty subsets of nodes in $\mathcal{V} \setminus \{i\}$. Thus we have $S_{\ref{alg.tightcheck},2}(N) = \sum_{\xi=1}^{N-1} \binom{N-1}{\xi}$, where $\xi = \vert \mathcal{A} \vert$.

\item (line 7-8) In this step, considering the set $\mathcal{A}$ determined in step (2), we find all the possible partitions $\mathcal{M}, \mathcal{L}$ in $\mathcal{V} \setminus \mathcal{A}$. Note that $s \in \mathcal{L}$ is specified in step (1). This can be done by finding all nonempty subsets $\mathcal{M}_\beta \subseteq \mathcal{V} \setminus (\mathcal{A}_\alpha \cup \{s\})$ that imposes $S_{\ref{alg.tightcheck},3}(N) = \sum_{\eta=1}^{N-\xi -1} \binom{N-\xi -1}{\eta}$ operations, where $\eta = \vert \mathcal{M} \vert$ and $\xi \in \{1,\ldots, N-1 \}$ comes from step (2).

\item (line 9-14) Incoming edges of every node (in worst case) in $\mathcal{M}$, obtained in step (3), from the nodes in $\mathcal{V} \setminus \mathcal{M}$ have to be checked to ensure that they do not have more than $f$ incoming edges from $\mathcal{A}$ and to find at least a node which has $f+1$ incoming edges from $\mathcal{L}$ or is directly connected to $s \in \mathcal{L}$. To check each incoming edge of a node, we deal with a test, which is a basic binary operation. Thus, this step needs $S_{\ref{alg.tightcheck},4}(N) = \eta (N-\eta)$ tests, where $\eta$ is specified in step (3). 
\end{enumerate}

\begin{figure}[t]
	\vspace{0cm}
	\hspace{.5cm}
	\def \svgscale{.35}
	\fontsize{7}{10}\selectfont
	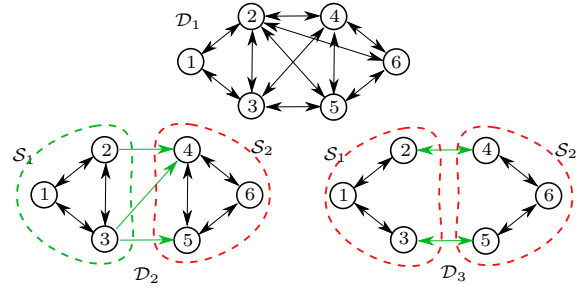
	\vspace{-.5cm}
    \caption{A variation of topology constraints on a graph with 7 nodes is illustrated. Digraph $\mathcal{D}_1$ is the most resilient constraint i.e. strongly $2$-robust; digraph $\mathcal{D}_2$ is 2-robust but it is not strongly 2-connected and digraph $\mathcal{D}_3$ is strongly $2$-connected while it is not $2$-robust.}
 \label{fig.robcon}
\end{figure}

Considering the aforementioned cascading steps, the total number of required basic operations to perform Algorithm~\ref{alg.tightcheck} is represented by
\begin{align} \label{eq.tigtcomp}
S_{\ref{alg.tightcheck}}(&N) = \nonumber \\ 
 & N \sum_{\xi=1}^{N-1} \left( \binom{N-1}{\xi}  \underbrace{\sum_{\eta=1}^{N-\xi -1} \binom{N-\xi -1}{\eta} \eta (N-\eta)}_{S_{\ref{alg.tightcheck}}^{\star}} \right).
\end{align}
Using Binomial Theorem\footnote{According to Binomial Theorem, we have $\sum_{\xi=0}^{N} \xi \binom{N}{\xi} = 2^{(N-1)}N$ and $\sum_{\xi=0}^{N} \xi^2 \binom{N}{\xi} = 2^{(N-2)}N(N+1)$, see \cite{boros2004irresistible}.}, $S_{\ref{alg.tightcheck}}^{\star}$ is simplified as $S_{\ref{alg.tightcheck}}^{\star} = 2^{(N-\xi -1)}(N-\xi -1)(3N-\xi)$ and \eqref{eq.tigtcomp} is reformed as
\begin{align} \label{eq.tigtcompnew}
S_{\ref{alg.tightcheck}}(N) =& 16 \times 3^{(N - 3)}N^3 - 3\times 2^{(N - 1)} N^3 \nonumber \\
&- 4\times 3^{(N - 2)} N^2 + 3\times 2^{(N - 1)} N^2 \nonumber \\
&- 4\times 3^{(N - 3)}N.
\end{align}

Finally, neglecting the lower order terms, the worst-case complexity of Algorithm~\ref{alg.tightcheck} is $\mathcal{O}(N^{3} 3^N)$.
\end{proof}

In large scale networks, time complexity of an algorithm really matters. The results of Theorems~\ref{th.coxrob} and \ref{th.coxtight} justify our selection of strong robustness as the required condition for our presented averaging consensus algorithms.
\subsection{Robustness vs. Connectivity} \label{sec. grrob}
In this section, we link our design to the previous existing literature of (resilient) consensus and averaging. In particular, Definitions~\ref{df.rob} and \ref{df.Srobw} have been used for resilient (non-average) consensus \cite{dibaji2018TAC} and resilient distributed estimation \cite{mitra2018secure}, respectively. Strong connectivity described in Definition~\ref{df.con}, on the other hand, is a topological condition for average consensus problems \cite{cai2012average}. Connections between these metrics are investigated as follows.

First, we present some properties of strongly $r$-robust graphs to deliver a better understanding of the robustness notions on graph topologies. We have to note that the strong robustness definition in \cite{zhang2012robustness} does not include $r \leq \lceil \frac{N}{2} \rceil$. To address this discrepency, consider the complete graph $\mathcal{K}_N$, which has the maximum degree of strong robustness, i.e. $r$. However, as every subset of $\mathcal{K}_N$ has some links from the nodes outside the subset, $r$ must be upper-bounded. For the sake of consistency with $r$-robust graphs, where $\mathcal{K}_N$ is $\lceil \frac{N}{2} \rceil$-robust and $r \leq \lceil \frac{N}{2} \rceil$ \cite{dibaji2018TAC}, the same bound for $r$ in strongly robust graphs is imposed. We also state how a strongly $r$-robust digraph is related to a strongly $r$-connected digraph.

\begin{proposition} \rm \label{prop.propsrob}
If the graph $\mathcal{D}$ with $N$ nodes is strongly $r$-robust, the following properties hold:
\begin{itemize}
\item[i)]  $\mathcal{D}$ is strongly $r$-connected. 
\item[ii)]  $\mathcal{D}$ is $r$-robust.
\item[iii)] $d_{\text{in}}(i) \geq r$ for all nodes $i \in \mathcal{V}$.
\item[iv)] For all $1 \leq r^\prime \leq r$, $\mathcal{D}$ is strongly $r^\prime$-robust.
\item[v)] Removal of $k$ incoming edges from each node of $\mathcal{D}$ makes it strongly $(r-k)$-robust.
\end{itemize}
\end{proposition}

\begin{proof}
i) We use contradiction to prove this. Suppose that the digraph $\mathcal{D}$ is not strongly $r$-connected but strongly $r$-robust; it means that there exists at least one node cut with at most $r-1$ nodes whose removal changes the digraph $\mathcal{D}$ from strongly $r$-connected, which is in the category $C_3$, to an $r$-connected graph in the category $C_0$, $C_1$, or $C_2$. Let $\mathcal{M}$ be the set of nodes removed in this node cut. Hence, $\mathcal{D} \setminus \mathcal{M}$ has at least $m$ strongly $1$-connected components where $m$ is a fixed number and $m \geq 2$.  Denote the components by $\mathcal{D}_1=(\mathcal{V}_1,\mathcal{E}_1), \mathcal{D}_2=(\mathcal{V}_2,\mathcal{E}_2), \ldots, \mathcal{D}_m=(\mathcal{V}_m,\mathcal{E}_m)$, where $\mathcal{D} \setminus \mathcal{M}$ cannot be in the category of $C_3$. Then, we prove that, in fact, $\mathcal{D} \setminus \mathcal{M}$ belongs to $C_3$. This part of the proof is also done with a contradiction. According to Def.~\ref{df.Srob}, $\mathcal{V}_1 \subset \mathcal{V}$ is $r$-reachable or contains a node which has incoming edges from every node outside of $\mathcal{D}_1$. However, as $\mathcal{M}$ has at most $r-1$ nodes, with another proof by contradiction, $\mathcal{D}_1$ has at least one incoming edge from the nodes in $\mathcal{D}_1^\prime = \bigcup_2^m \mathcal{D}_m$, $m \geq 2$. Likewise, $\mathcal{D}_1^\prime$ has one incoming edge from $\mathcal{D}_1$. This means that $\mathcal{D} \setminus \mathcal{M}$ does not have $m$ distinct strongly connected subgraphs. Hence, we have at most $m-1$ strongly connected components. By repeating this process and mathematical induction, we see that $m$ cannot be more than 1 which is a contradiction to the existence of such an $\mathcal{M}$. In fact, $\mathcal{D}\setminus \mathcal{M}$ is at least strongly $1$-connected, i.e. belongs to $C_3$.

ii) Assume that $\mathcal{D}$ is strongly $r$-robust and $\mathcal{S}$ is an arbitrary subset of $\mathcal{V}$. Then $ \forall \mathcal{S} \subset \mathcal{V}$, $\mathcal{S} \neq \emptyset$, $\exists i \in \mathcal{S}$, such that whether $\vert \mathcal{N}_i^- \setminus \mathcal{S} \vert \geq r$ or $\mathcal{N}_i^- \setminus \mathcal{S} = \mathcal{V} \setminus \mathcal{S}$. We have to prove that $\forall \mathcal{S}_1,\mathcal{S}_2 \neq \emptyset$, $\mathcal{S}_1 \cap \mathcal{S}_2 \neq \emptyset$, whether $\mathcal{S}_1$ is $r$-reachable or $\mathcal{S}_2$ is $r$-reachable. First, assume that $\vert \mathcal{S}_1 \vert \leq \lfloor \frac{N}{2} \rfloor$. Then, $\exists i \in \mathcal{S}_1$, where $\vert \mathcal{N}_i^- \setminus \mathcal{S}_1 \vert \geq r$ or $\mathcal{N}_i^- \setminus \mathcal{S}_1  = \mathcal{V} \setminus \mathcal{S}_1$. However, note that $\vert \mathcal{V} \setminus \mathcal{S}_1 \vert \geq \lfloor \frac{N}{2} \rfloor$ while $r \leq \lfloor \frac{N}{2} \rfloor$. Thus, $\vert \mathcal{N}_i^- \setminus \mathcal{S}_1 \vert \geq r$ which means that $\mathcal{S}_1$ is $r$-reachable. Now, consider the case that $\vert \mathcal{S}_1 \vert \geq \lceil \frac{N}{2} \rceil$. Then, there is a set $\mathcal{S}_2 \subseteq \mathcal{V} \setminus \mathcal{S}_1$, where $\vert \mathcal{S}_2 \vert \leq \vert \mathcal{V} \setminus \mathcal{S}_1 \vert \leq \lfloor \frac{N}{2} \rfloor$. With a similar argument to the previous case, we conclude that $\mathcal{S}_2$ is $r$-reachable. Therefore, in both cases one of the two partition sets is $r$-reachable.

The remaining properties iii)-v) can be shown easily and their proofs are omitted. 
\end{proof}

 Note that we neither did not make use of Def.~\ref{df.Srobw} \cite{mitra2018secure} in this paper but some of the properties in Proposition \ref{prop.propsrob} can be directly applied to them as each graph which is strongly $r$-robust with respect to each of its nodes is in fact strongly $r$-robust. Furthermore, the relation between an $r$-robust graph and a strongly $r$-connected graph is still an open problem to investigate. In fact, there are some graphs which are $r$-robust but are not strongly $r$-connected and vice versa. Three examples for the case $r=2$ are illustrated in Fig.~\ref{fig.robcon}. Digraph $\mathcal{D}_1$ is strongly 2-robust. Digraph $\mathcal{D}_2$ is 2-robust but it is not strongly 2-connected. In $\mathcal{D}_2$, consider two subsets of nodes $\mathcal{S}_1 = \{1,2,3\}$ and $\mathcal{S}_2=\{4,5,6\}$. There exists paths only from nodes in $\mathcal{S}_1$ to nodes in $\mathcal{S}_2$ but the inverse is not true. Therefore, the digraph cannot be strongly 2-connected while it is 2-robust since $\mathcal{S}_1$ is 2-reachable. However, in digraph $\mathcal{D}_3$ which is a strongly 2-connected digraph (circular graphs are strongly 2-connected), none of the two subsets $\mathcal{S}_1$ and $\mathcal{S}_2$ are 2-reachable, thus it is not 2-robust.

We also infer the following corollary as another bridge between connectivity and robustness.

\begin{corollary} \rm \label{cor.1}
A digraph $\mathcal{D}$ is strongly 1-connected if and only if it is strongly 1-robust. 
\end{corollary}

As a result, the following necessary and sufficient condition holds on a network with no adversarial node to achieve average consensus, which is well known in the literature \cite{cai2012average}.

\begin{corollary} \rm \label{cor.2}
Network $\mathcal{D}$ with no adversarial nodes using (A)SABA and update rules (2) or (3), achieves average consensus if and only if it is strongly connected.
\end{corollary}

\begin{proof}
By Corollary~\ref{cor.1} and Theorem \ref{th.sum}, it immediately follows that strong connectivity is the sufficient condition for the network $\mathcal{D}$ to achieve average consensus with no avdersarial nodes, or $f=0$. Also, if network $\mathcal{D}$ has no adversarial nodes, at least, every two nodes of $\mathcal{D}$ has to be mutually reachable (or $\mathcal{D}$ has to be strongly connected), thus each node can retrieve all the initial values and calculate the average state value.
\end{proof}


\begin{figure*}[t]
\minipage{0.247\textwidth}
	\vspace{0cm}
	\def \svgscale{.33}
	\fontsize{7}{10}\selectfont
	\hspace{0cm}
\begingroup%
  \makeatletter%
  \providecommand\color[2][]{%
    \errmessage{(Inkscape) Color is used for the text in Inkscape, but the package 'color.sty' is not loaded}%
    \renewcommand\color[2][]{}%
  }%
  \providecommand\transparent[1]{%
    \errmessage{(Inkscape) Transparency is used (non-zero) for the text in Inkscape, but the package 'transparent.sty' is not loaded}%
    \renewcommand\transparent[1]{}%
  }%
  \providecommand\rotatebox[2]{#2}%
  \newcommand*\fsize{\dimexpr\f@size pt\relax}%
  \newcommand*\lineheight[1]{\fontsize{\fsize}{#1\fsize}\selectfont}%
  \ifx\svgwidth\undefined%
    \setlength{\unitlength}{409.80759989bp}%
    \ifx\svgscale\undefined%
      \relax%
    \else%
      \setlength{\unitlength}{\unitlength * \real{\svgscale}}%
    \fi%
  \else%
    \setlength{\unitlength}{\svgwidth}%
  \fi%
  \global\let\svgwidth\undefined%
  \global\let\svgscale\undefined%
  \makeatother%
  \begin{picture}(1,0.43053273)%
    \lineheight{1}%
    \setlength\tabcolsep{0pt}%
    \put(0,0){\includegraphics[width=\unitlength,page=1]{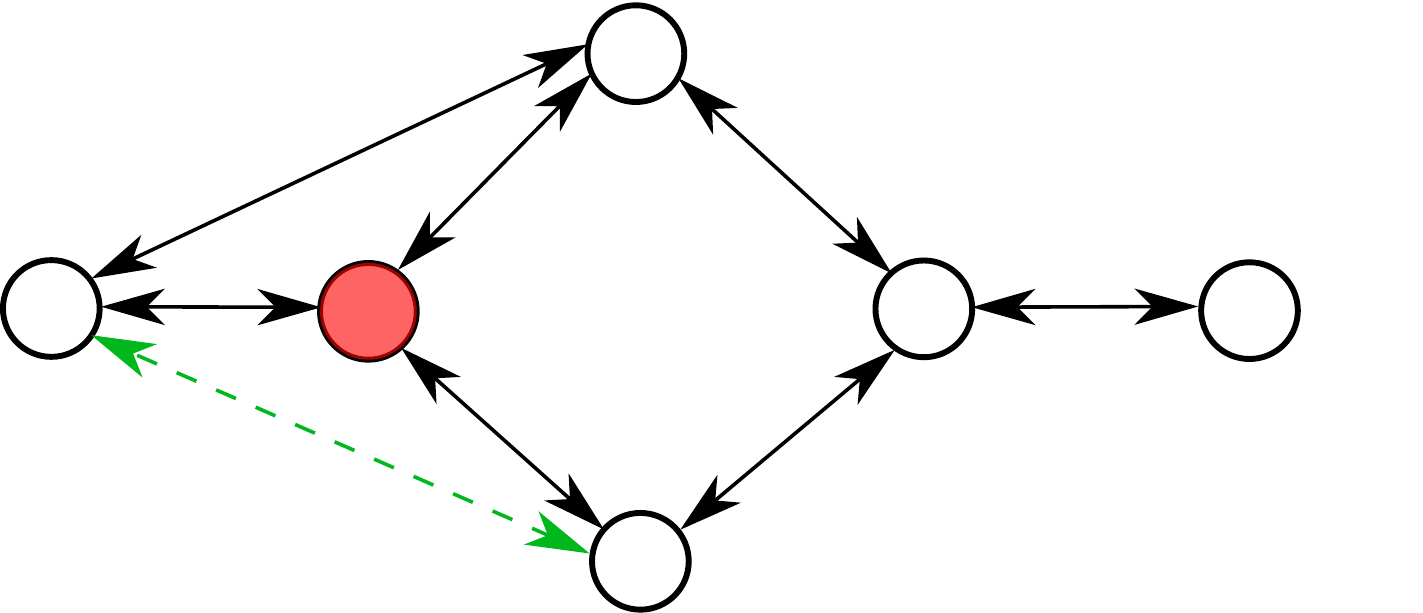}}%
    \put(0.4347139,0.37419916){\color[rgb]{0,0,0}\makebox(0,0)[lt]{\lineheight{1.25}\smash{\begin{tabular}[t]{l}$1$\end{tabular}}}}%
    \put(0.63259986,0.19860911){\color[rgb]{0,0,0}\makebox(0,0)[lt]{\lineheight{1.25}\smash{\begin{tabular}[t]{l}$2$\end{tabular}}}}%
    \put(0.43471395,0.01909864){\color[rgb]{0,0,0}\makebox(0,0)[lt]{\lineheight{1.25}\smash{\begin{tabular}[t]{l}$3$\end{tabular}}}}%
    \put(0.2427827,0.19786575){\color[rgb]{0,0,0}\makebox(0,0)[lt]{\lineheight{1.25}\smash{\begin{tabular}[t]{l}$4$\end{tabular}}}}%
    \put(0.02333472,0.19664892){\color[rgb]{0,0,0}\makebox(0,0)[lt]{\lineheight{1.25}\smash{\begin{tabular}[t]{l}$5$\end{tabular}}}}%
    \put(0.86531413,0.19476299){\color[rgb]{0,0,0}\makebox(0,0)[lt]{\lineheight{1.25}\smash{\begin{tabular}[t]{l}$6$\end{tabular}}}}%
    \put(0,0){\includegraphics[width=\unitlength,page=2]{Fig2.pdf}}%
  \end{picture}%
\endgroup%

	\vspace{0cm}
    \subcaption{} \label{fig2} 
\endminipage\hfill
\hspace{0mm}
\minipage{0.247\textwidth}
	\includegraphics[scale=0.34]{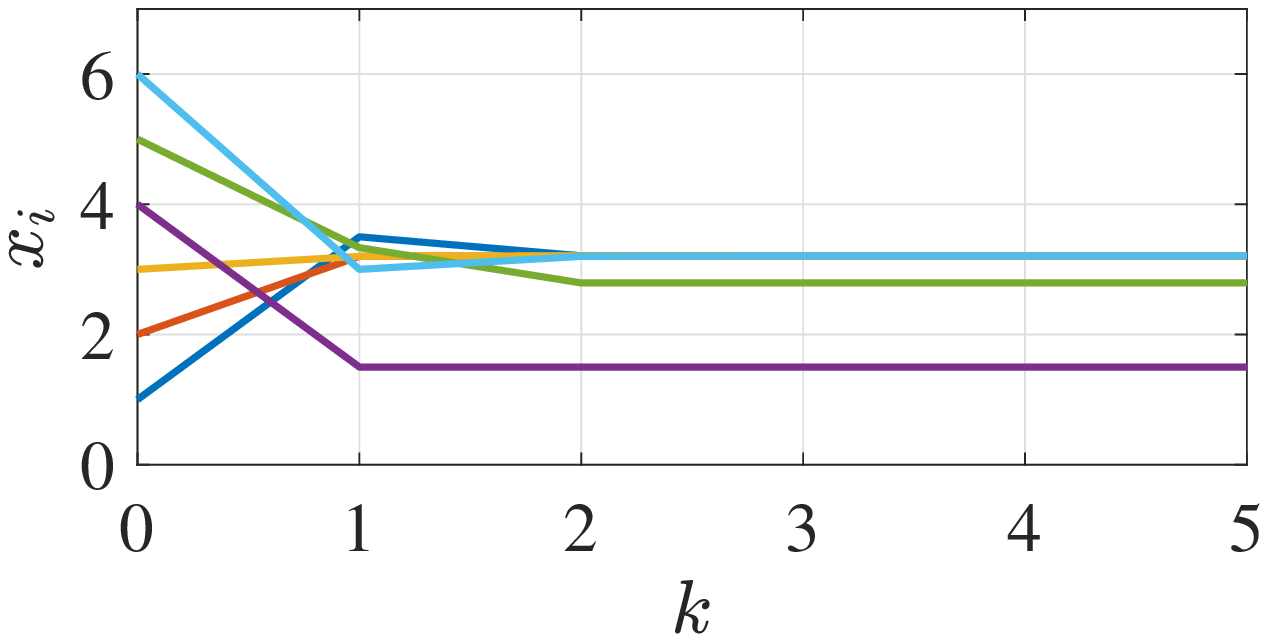}
	\subcaption{}\label{fig.SABAF}
\endminipage\hfill
\minipage{0.247\textwidth}%
	\includegraphics[scale=0.34]{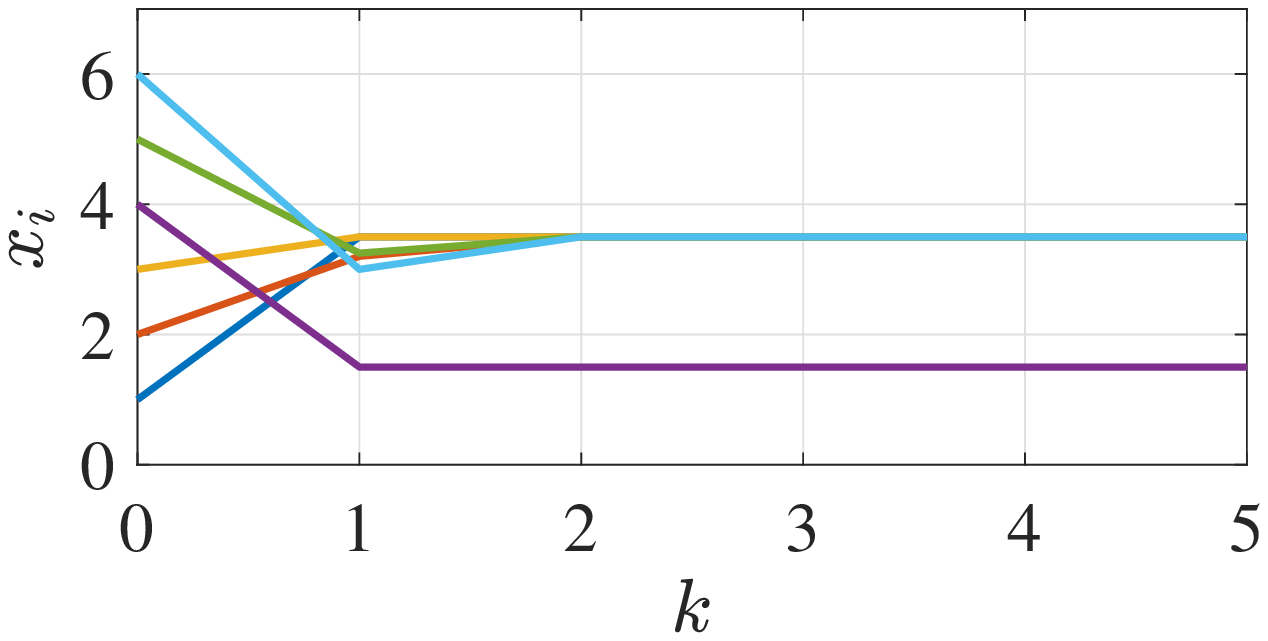}
	\subcaption{} \label{fig.SABAW}
\endminipage\hfill
\minipage{0.247\textwidth}%
	\includegraphics[scale=0.34]{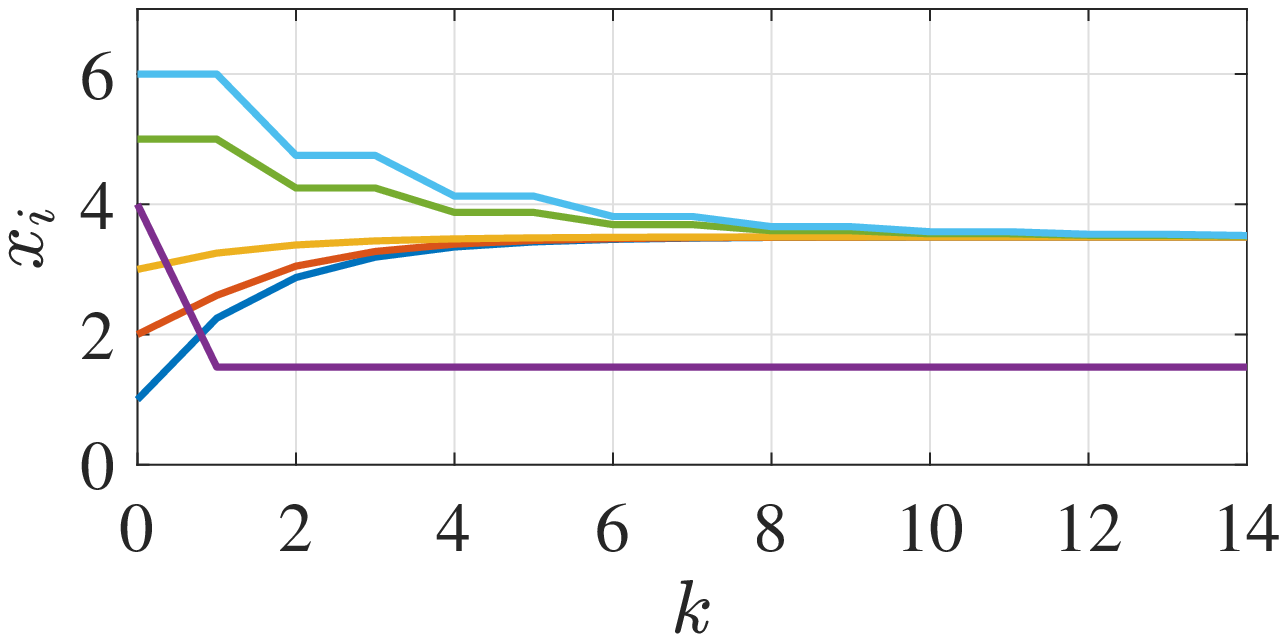}
	\subcaption{}\label{fig.SABAW_A}
\endminipage

\caption{A strongly 3-robust graph with 6 nodes. Node 4 is the adversarial node (a). If we remove the edge between the nodes 3 and 5, the network is not strongly ($2f+1$)-robust anymore and fails to reach the average consensus by perfoming SABA/ASABA (b). Otherwise, the network reaches average consensus using SABA (c) and ASABA (d).}\label{fig.SIMMM}
\end{figure*} 
\section{Simulation Example}\label{Sect: SimulationExample}
In this section, we provide a simulation example to show the effectiveness of the (A)SABA in the presence of Byzantine adversaries. 

Consider the six-node network illustrated in Fig.~\ref{fig2}, which enjoys strong ($2f+1$)-robustness with $f=1$. We assign the safe initial state values as $x_i[0]=i, \ i=1,2, \ldots, 6$, which turn into the average value $x_a=3.5$. Node 4 is chosen as the adversarial node. It begins and continues broadcasting false values $m_n^4[k]=1.5$, $n \in \{ 1,2,\ldots,\bar{N} \}$ of its neighbours for $k \geq 1$. Every other node since $k=1$ performs the SABA, updates its state value, and compute the average with the filter gain $\epsilon_i = 0$). The network loses its robustness if we remove the edge between nodes 3 and 5, leading to failure of the averaging algorithm and this can be seen in Fig.~\ref{fig.SABAF}. 
On the contrary, if the network of Fig.~\ref{fig2} includes the edge between the node 3 and 5, it is strongly ($2f+1$)-robust and achieves the average consensus as illustrated in Fig.~\ref{fig.SABAW}, where the regular nodes achieve $x_i[\bar{K}]=x_a, \ i=1,2,3,5,6$.

Furthermore, Fig.~\ref{fig.SABAW_A} represents the states of the network, where nodes 5 and 6 execute ASABA at $k=0,2,4,\ldots$ and all the nodes update using the asynchronous update rule (\ref{eq.asupd}). We set the filter gain as $\epsilon_i=0.5$ and the smoothing effect of the filters on the state variables is noticeable.

\section{Conclusion}\label{Sec: Conclusion}

In this paper, we have presented fully distributed algorithms for the problem of average consensus in the presence of Byzantine agents. We have shown that the convergence of presented algorithms on synchronous and asynchronous networks in the presence $f$-local adversaries relies on a connectivity measure of graphs, which is called strong robustness and is more restrictive than the conventional node connectivities. We also justified that strong robustness is the most proper topology constraint for the resilient distributed averaging against similar graph conditions proposed in the literature. We presented an analysis in terms of time complexity needed for checking these topology constraints. Simulation results have validated the proposed algorithms and convergence conditions. In the future, we will extend our results further to fit real applications.

%
%
%
%

\section*{Acknowledgement}
We offer our sincerest gratitude to Prof. Hideaki Ishii for the time he dedicated to us for useful discussions on the topic as well as his technical comments which significantly helped us to improve the quality of this paper. 

\bibliography{Avrg}

\begin{thebibliography}{10}
\providecommand{\url}[1]{#1}
\csname url@samestyle\endcsname
\providecommand{\newblock}{\relax}
\providecommand{\bibinfo}[2]{#2}
\providecommand{\BIBentrySTDinterwordspacing}{\spaceskip=0pt\relax}
\providecommand{\BIBentryALTinterwordstretchfactor}{4}
\providecommand{\BIBentryALTinterwordspacing}{\spaceskip=\fontdimen2\font plus
\BIBentryALTinterwordstretchfactor\fontdimen3\font minus
  \fontdimen4\font\relax}
\providecommand{\BIBforeignlanguage}[2]{{%
\expandafter\ifx\csname l@#1\endcsname\relax
\typeout{** WARNING: IEEEtran.bst: No hyphenation pattern has been}%
\typeout{** loaded for the language `#1'. Using the pattern for}%
\typeout{** the default language instead.}%
\else
\language=\csname l@#1\endcsname
\fi
#2}}
\providecommand{\BIBdecl}{\relax}
\BIBdecl

\bibitem{dibaji2019systems}
S.~M. Dibaji, M.~Pirani, D.~B. Flamholz, A.~M. Annaswamy, K.~H. Johansson, and
  A.~Chakrabortty, ``A systems and control perspective of {CPS} security,''
  \emph{Annual Reviews in Control}, vol.~47, pp. 394 -- 411, 2019.

\bibitem{moreau2005stability}
L.~Moreau, ``Stability of multiagent systems with time-dependent communication
  links,'' \emph{IEEE Trans. on Autom. Control}, vol.~50, no.~2, pp. 169--182,
  2005.

\bibitem{jadbabaie2003coordination}
A.~Jadbabaie, J.~Lin, and A.~S. Morse, ``Coordination of groups of mobile
  autonomous agents using nearest neighbor rules,'' \emph{IEEE Trans. on Autom.
  Control}, vol.~48, no.~6, pp. 988--1001, 2003.

\bibitem{cai2012average}
K.~Cai and H.~Ishii, ``Average consensus on general strongly connected
  digraphs,'' \emph{Automatica}, vol.~48, no.~11, pp. 2750--2761, 2012.

\bibitem{kashyap2007quantized}
A.~Kashyap, T.~Ba{\c{s}}ar, and R.~Srikant, ``Quantized consensus,''
  \emph{Automatica}, vol.~43, no.~7, pp. 1192--1203, 2007.

\bibitem{dibaji2018TAC}
S.~M. Dibaji, H.~Ishii, and R.~Tempo, ``Resilient randomized quantized
  consensus,'' \emph{IEEE Trans. on Autom. Control}, vol.~63, no.~8, pp.
  2508--2522, 2018.

\bibitem{Hadjicostis}
A.~D. Dominguez-Garcia and C.~N. Hadjicostis, ``Coordination and control of
  distributed energy resources for provision of ancillary services,'' in
  \emph{Proc. of International Conference on Smart Grid Communications}, 2010,
  pp. 537--542.

\bibitem{mo2017privacy}
Y.~Mo and R.~M. Murray, ``Privacy preserving average consensus,'' \emph{IEEE
  Trans. on Autom. Control}, vol.~62, no.~2, pp. 753--765, 2017.

\bibitem{sundaram2011distributed}
S.~Sundaram and C.~N. Hadjicostis, ``Distributed function calculation via
  linear iterative strategies in the presence of malicious agents,'' \emph{IEEE
  Trans. on Autom. Control}, vol.~56, no.~7, pp. 1495--1508, 2011.

\bibitem{pasqualetti2012consensus}
F.~Pasqualetti, A.~Bicchi, and F.~Bullo, ``Consensus computation in unreliable
  networks: A system theoretic approach,'' \emph{IEEE Trans. on Autom.
  Control}, vol.~57, no.~1, pp. 90--104, 2012.

\bibitem{meskin2009actuator}
N.~Meskin and K.~Khorasani, ``Actuator fault detection and isolation for a
  network of unmanned vehicles,'' \emph{IEEE Trans. on Autom. control},
  vol.~54, no.~4, pp. 835--840, 2009.

\bibitem{azadmanesh2002asynchronous}
M.~Azadmanesh and R.~Kieckhafer, ``Asynchronous approximate agreement in
  partially connected networks,'' \emph{International Journal of Parallel and
  Distributed Systems and Networks}, vol.~5, no.~1, pp. 26--34, 2002.

\bibitem{bouzid2010optimal}
Z.~Bouzid, M.~G. Potop-Butucaru, and S.~Tixeuil, ``Optimal
  {B}yzantine-resilient convergence in uni-dimensional robot networks,''
  \emph{Theoretical Computer Science}, vol. 411, no. 34-36, pp. 3154--3168,
  2010.

\bibitem{vaidya2012iterative}
N.~H. Vaidya, L.~Tseng, and G.~Liang, ``Iterative approximate {B}yzantine
  consensus in arbitrary directed graphs,'' in \emph{Proc. of the ACM symposium
  on Principles of distributed computing}.\hskip 1em plus 0.5em minus
  0.4em\relax ACM, 2012, pp. 365--374.

\bibitem{leblanc2013resilient}
H.~J. LeBlanc, H.~Zhang, X.~Koutsoukos, and S.~Sundaram, ``Resilient asymptotic
  consensus in robust networks,'' \emph{IEEE Journal on Selected Areas in
  Communications}, vol.~31, no.~4, pp. 766--781, 2013.

\bibitem{dibaji2015consensus}
S.~M. Dibaji and H.~Ishii, ``Consensus of second-order multi-agent systems in
  the presence of locally bounded faults,'' \emph{Systems \& Control Letters},
  vol.~79, pp. 23--29, 2015.

\bibitem{zhang2012robustness}
H.~Zhang and S.~Sundaram, ``Robustness of information diffusion algorithms to
  locally bounded adversaries,'' in \emph{Proc. of American Control
  Conference}, 2012, pp. 5855--5861.

\bibitem{koo2004broadcast}
C.-Y. Koo, ``Broadcast in radio networks tolerating {B}yzantine adversarial
  behavior,'' in \emph{Proc. of ACM Symposium on Principles of Distributed
  Computing}, 2004, pp. 275--282.

\bibitem{pelc2005locallybounded}
A.~Pelc and D.~Peleg, ``Broadcasting with locally bounded {B}yzantine faults,''
  \emph{Information Processing Letters}, vol.~93, no.~3, pp. 109--115, 2005.

\bibitem{bhandari2009reliable}
V.~Bhandari and N.~H. Vaidya, ``Reliable broadcast in radio networks with
  locally bounded failures,'' \emph{IEEE Trans. on Parallel and Distributed
  Systems}, vol.~21, no.~6, pp. 801--811, 2009.

\bibitem{ichimura2010new}
A.~Ichimura and M.~Shigeno, ``A new parameter for a broadcast algorithm with
  locally bounded {B}yzantine faults,'' \emph{Information processing letters},
  vol. 110, no. 12-13, pp. 514--517, 2010.

\bibitem{tseng2015broadcast}
L.~Tseng, N.~Vaidya, and V.~Bhandari, ``Broadcast using certified propagation
  algorithm in presence of {B}yzantine faults,'' \emph{Information Processing
  Letters}, vol. 115, no.~4, pp. 512--514, 2015.

\bibitem{dibaji2019resilient}
S.~M. Dibaji, M.~Safi, and H.~Ishii, ``Resilient distributed averaging,'' in
  \emph{Proc. of American Control Conference}, 2019, pp. 96--101.

\bibitem{mitra2018secure}
A.~Mitra and S.~Sundaram, ``Secure distributed state estimation of an lti
  system over time-varying networks and analog erasure channels,'' in
  \emph{2018 Annual American Control Conference (ACC)}.\hskip 1em plus 0.5em
  minus 0.4em\relax IEEE, 2018, pp. 6578--6583.

\bibitem{geller1971connectivity}
D.~Geller and F.~Harary, ``Connectivity in digraphs,'' in \emph{Recent Trends
  in Graph Theory}.\hskip 1em plus 0.5em minus 0.4em\relax Springer, 1971, pp.
  105--115.

\bibitem{gross2005graph}
J.~L. Gross and J.~Yellen, \emph{{Graph Theory and its Applications}}.\hskip
  1em plus 0.5em minus 0.4em\relax CRC Press, 2005.

\bibitem{brown2004smoothing}
R.~G. Brown, \emph{{Smoothing, Forecasting and Prediction of Discrete Time
  Series}}.\hskip 1em plus 0.5em minus 0.4em\relax Courier Corporation, 2004.

\bibitem{west1996introduction}
D.~B. West, \emph{Introduction to {G}raph {T}heory}.\hskip 1em plus 0.5em minus
  0.4em\relax Prentice hall Upper Saddle River, NJ, 1996, vol.~2.

\bibitem{knuth1976big}
D.~E. Knuth, ``Big omicron and big omega and big theta,'' \emph{ACM {SIGACT}
  News}, vol.~8, no.~2, pp. 18--24, 1976.

\bibitem{boros2004irresistible}
G.~Boros and V.~Moll, \emph{Irresistible integrals: symbolics, analysis and
  experiments in the evaluation of integrals}.\hskip 1em plus 0.5em minus
  0.4em\relax Cambridge University Press, 2004.

\end{thebibliography}
\bibliographystyle{IEEEtran}

\end{document}